\documentclass[a4paper, USenglish, cleveref, nameinlink, autoref, pagebackref]{lipics-v2019}

\newif\ifpaper
\papertrue
\ifpaper
  \nolinenumbers
  \hideLIPIcs
\fi

\usepackage{tikz}
\usepackage[utf8]{inputenc} 
\usepackage{amsmath}
\usepackage[disable]{todonotes}
\usepackage{dsfont}
\usepackage{nicefrac}

\usepackage{multirow,booktabs}

\usepackage[noend]{algpseudocode}
\usepackage{algorithm}

\usepackage[numbers,sort,sectionbib]{natbib}

\theoremstyle{plain}
\newtheorem{observation}[theorem]{Observation}

\crefformat{footnote}{#2\footnotemark[#1]#3}

\title{Detecting and Enumerating Small Induced Subgraphs in $c$-Closed Graphs}

\titlerunning{Detecting and Enumerating Small Induced Subgraphs in $c$-Closed Graphs}

\newcommand{\tuaddress}{Algorithmics and Computational Complexity, Faculty IV, TU Berlin, Germany}

\author
{Tomohiro Koana}
{\tuaddress}
{tomohiro.koana@tu-berlin.de}
{https://orcid.org/0000-0002-8684-0611} 
{Supported by the DFG project FPTinP (NI 369/16).} 

\author
{Andr\'{e} Nichterlein}
{\tuaddress}
{andre.nichterlein@tu-berlin.de}
{https://orcid.org/0000-0001-7451-9401} 
{} 

\authorrunning{Tomohiro Koana, André Nichterlein}

\Copyright{Tomohiro Koana, André Nichterlein}



\EventEditors{John Q. Open and Joan R. Access}
\EventNoEds{2}
\EventLongTitle{42nd Conference on Very Important Topics (CVIT 2016)}
\EventShortTitle{CVIT 2016}
\EventAcronym{CVIT}
\EventYear{2016}
\EventDate{December 24--27, 2016}
\EventLocation{Little Whinging, United Kingdom}
\EventLogo{}
\SeriesVolume{42}
\ArticleNo{23}

\Crefname{theorem}{Theorem}{Theorems}
\crefname{theorem}{Thm.}{Thms.}
\Crefname{observation}{Observation}{Observations}
\crefname{observation}{Obs.}{Obs.}
\Crefname{corollary}{Corollary}{Corollaries}
\crefname{corollary}{Cor.}{Cor.}

\algnewcommand\Continue{\textbf{continue}}

\newcommand{\co}[1]{\overline{#1}}

\begin{document}

\maketitle

\begin{abstract}
	\citeauthor{FRSWW20} [SIAM J. Comp. 2020] introduced a new parameter, called $c$-closure, for a parameterized study of clique enumeration problems.
	A graph~$G$ is $c$-closed if every pair of vertices with at least~$c$ common neighbors is adjacent.
	The $c$-closure of~$G$ is the smallest~$c$ such that~$G$ is $c$-closed.
	We systematically explore the impact of $c$-closure on the computational complexity of detecting and enumerating small induced subgraphs.
	More precisely, for each graph~$H$ on three or four vertices, we investigate parameterized polynomial-time algorithms for detecting~$H$ and for enumerating all occurrences of~$H$ in a given $c$-closed graph.
\end{abstract}

\ifpaper\else\newpage\fi

\section{Introduction}

Detecting and enumerating a fixed subgraph~$H$ in a given host graph~$G$ is an important and well-studied graph problem.
Even the special cases for small subgraphs~$H$ have many applications e.\,g.\ in the analysis of protein--protein networks~\cite{PCJ06} or of social networks~\cite{SPRH06}.

We focus on the problem variants where~$H$ has three or four vertices, resulting in 15 problem variants.
Out of these 15 candidates for~$H$, only three are known to be detectable in linear time: a path on three~$P_3$ or four~$P_4$ vertices~\cite{CPS85} and the complement of a~$P_3$ (an edge plus an isolated vertex).
For the remaining 12 subgraphs the (theoretically) fastest known algorithms are based on fast matrix multiplication~\cite{WWWY15,IR78} and mostly run in~$\widetilde{O}(n^\omega)$\footnote{The $\widetilde{O}$-notation suppresses polylogarithmic factors. Here~$O(n^\omega)$ is the time to multiply two~$n \times n$-matrices; it is known that~$\omega < 2.3728639$~\cite{Gal14}.} time ($O(n^{3.257})$ for clique and independent set on four vertices).
However, the fast matrix multiplication is not practical due to its large overhead.
We will thus focus on ``combinatorial'' algorithms.
Although this term is not well-defined, it is usually used to denote algorithms without any use of fast matrix multiplication.
These algorithms are often more efficient in practice.

Finding a combinatorial algorithm that detects a triangle in $O(n^{3 - \varepsilon})$ for $\varepsilon > 0$ seems challenging to an extent that it was conjectured to not exist~\cite{AVW14,WW18}. 
To circumvent this difficulty, we follow the spirit of ``parameterization for polynomial-time solvable problems'' (also referred to as ``FPT~in~P'')~\cite{GMN17}. 
Our parameter of choice is the recently defined $c$-closure which captures a natural property found often in social networks~\cite{FRSWW20}: two vertices with many common neighbors tend to be adjacent.
More formally, the $c$-closure of a graph is the smallest integer~$c$ such that any two non-adjacent vertices have less than~$c$ common neighbors.

An advantage of the $c$-closure is that its reasonably small in social networks with thousands of vertices~\cite{FRSWW20}.
We provide FPT~in~P algorithms with a small polynomial dependency on~$c$, thus parameter-values which are prohibitively high for exponential-time algorithms are still be acceptable in our setting.

Besides induced subgraph detection algorithms we also investigate the enumeration problems.
Here, we settle for all but four out of the 15 subgraphs the complexity on $c$-closed graphs; see \Cref{tab:overview} for an overview on our results and existing work.

\newcommand{\tight}{\ensuremath{^T}}
\begin{table}[t!]
	\centering
	\caption{
		Overview for combinatorial algorithms. The algorithms of Williams \cite{WWWY15} are randomized; we use~$\omega = 3$ here to account for combinatorial algorithms. 
		The graphs are sorted first by the number of vertices and second by the number of edges.
		For each graph the first row refers to the running times on general graphs and the subsequent rows on $c$-closed graphs.
		Running time lower bounds (indicated by $\Theta(\cdot)$ and~$\Omega(\cdot)$) are mostly based on the number of possible occurrences of the subgraph in general graphs (first row) and $c$-closed graphs for constant~$c$ (subsequent rows).
		Results without reference are trivial and/or folklore and are discussed in \cref{sec:three-vertex,sec:four-vertex} for completeness.
	}
	\label{tab:overview}
	\begin{tabular}{lcrlrl}
		\toprule       
				& & \multicolumn{2}{c}{detection}					& \multicolumn{2}{c}{enumeration}\\ 
				\midrule
		\multirow{8}{*}{\rotatebox[origin=c]{90}{\cref{sec:three-vertex}}} & 
		\multirow{2}{*}{$\co{K_3}$} %
				& $O(n^3)$				&					& \multirow{2}{*}{$\Theta(n^3)$} &  \multirow{2}{*}{[\cref{obs:n^k-subgraphs}]} \\
		&		& $O(m + c^3)$ 			& \cite{KKS20}			& & \\ \cmidrule(lr){3-4} \cmidrule(lr){5-6} 
		& \multirow{1}{*}{$\co{P_3}$} %
				& \multirow{1}{*}{$\Theta(n+m)$} &					& \multirow{1}{*}{$\Theta(nm)$} & \multirow{1}{*}{[\cref{obs:n^k-subgraphs}]} \\
		\cmidrule(lr){3-4} \cmidrule(lr){5-6} 
		&\multirow{2}{*}{$P_3$} %
				& \multirow{2}{*}{$\Theta(n + m)$ }&				& $\Theta(nm)$ & \\
		&		& 						&					& $O(cn^2 + m^{\nicefrac{3}{2}})$ & [\cref{thm:enump3}] \\ \cmidrule(lr){3-4} \cmidrule(lr){5-6} 
		&\multirow{2}{*}{Triangle} %
				& $O(m^{\nicefrac{3}{2}})$ 	& \cite{IR78}			& \multirow{2}{*}{$\Theta(m^{\nicefrac{3}{2}})$} & \multirow{2}{*}{\cite{IR78}} \\
		&		& $O(cn^2)$, $O(c^{\nicefrac{1}{3}} m^{\nicefrac{4}{3}})$ & [\cref{thm:findk3,thm:findk3-2}] & & \\ \midrule
		\multirow{25}{*}{\rotatebox[origin=c]{90}{\cref{sec:four-vertex}}} & 
		\multirow{2}{*}{$\co{K_4}$} %
				& $O(n^4)$ 				&					& \multirow{2}{*}{$\Theta(n^4)$} & \multirow{2}{*}{[\cref{obs:n^k-subgraphs}]} \\ 
		&		& $O(m + c^4)$ 			& \cite{KKS20}			& & \\ \cmidrule(lr){3-4} \cmidrule(lr){5-6} 
		&\multirow{2}{*}{co-diamond} %
				& $O(m^{\nicefrac{3}{2}})$ 	& \cite{WWWY15}		& \multirow{2}{*}{$\Theta(n^2 m)$} & \multirow{2}{*}{[\cref{obs:n^k-subgraphs}]} \\ 
		&		& $O(m + c^2n)$ 			& [\cref{thm:codiamond}]	& & \\ \cmidrule(lr){3-4} \cmidrule(lr){5-6} 
		&\multirow{2}{*}{co-paw} %
				& $O(m^{\nicefrac{3}{2}})$ 	& \cite{WWWY15}		& $\Theta(n m^2)$, $\Theta(n^2 m)$ & \\
		&		& $O(m + c^3)$ 			& [\cref{cor:findcopaw}]	& $O(c n^3)$, $\Omega(n^3)$ & [\cref{thm:enumh}] \\ \cmidrule(lr){3-4} \cmidrule(lr){5-6} 
		&\multirow{2}{*}{co-square} %
				& $O(n^3)$, $O(m^{\nicefrac{11}{7}})$ & \cite{WWWY15}	& \multirow{2}{*}{$\Theta(m^2)$} & \multirow{2}{*}{[\cref{obs:n^k-subgraphs}]} \\ 
		&		& $O(m + c^{\nicefrac{44}{7}})$ & [\cref{thm:cosquare}]& & \\ \cmidrule(lr){3-4} \cmidrule(lr){5-6} 
		&\multirow{3}{*}{$P_4$} %
				& \multirow{3}{*}{$\Theta(n+m)$} & \multirow{3}{*}{\cite{CPS85}}		& $\Theta(m^2)$ & \\
		&		&						&					& $O(cnm)$ & [\cref{obs:enumP4}] \\ 
		&		&						&					& $\Omega(n^{2.5})$, $\Omega(m^{2})$ & [\cref{thm:p4c4lowerbounds}] \\ \cmidrule(lr){3-4} \cmidrule(lr){5-6} 
		&\multirow{2}{*}{claw} %
				& $O(m^{\nicefrac{3}{2}})$ 	& \cite{WWWY15}		& $\Theta(n^2 m)$ & \\
		&		& \multicolumn{2}{c}{?}							& $O(c n^3)$, $\Omega(n^3)$ & [\cref{thm:enumh}] \\ \cmidrule(lr){3-4} \cmidrule(lr){5-6} 
		&\multirow{2}{*}{co-claw} %
				& $O(m^{\nicefrac{3}{2}})$ 	& \cite{WWWY15}		& \multirow{2}{*}{$\Theta(n m^{\nicefrac{3}{2}})$} & \multirow{2}{*}{[\cref{obs:n^k-subgraphs}]} \\ 
		&		& \multicolumn{2}{c}{?}							& &\\ \cmidrule(lr){3-4} \cmidrule(lr){5-6} 
		&\multirow{2}{*}{paw} %
				& $O(m^{\nicefrac{3}{2}})$ 	& \cite{WWWY15}		& $\Theta(m^2)$ & \\
		&		& $O(cn^2)$, $O(c^{\nicefrac{1}{3}} m^{\nicefrac{4}{3}})$ & [\cref{cor:findpaw}]  & $O(c n m)$, $\Omega(n^3)$ & [\cref{obs:enumP4}] \\ \cmidrule(lr){3-4} \cmidrule(lr){5-6} 
		&\multirow{3}{*}{square} %
				& $O(n^3)$, $O(m^{\nicefrac{11}{7}})$ & \cite{WWWY15}	& $\Theta(m^2)$ \\
		&		& $O(c m^{\nicefrac{3}{2}})$ 	& [\cref{thm:enumc4}]	& $O(c^2 n^2 + m^{\nicefrac{3}{2}})$ & [\cref{thm:enumh}] \\
		&		& 						&					& $O(c m^{\nicefrac{3}{2}})$ & [\cref{thm:enumc4}] \\ 
		&		&						&					& $\Omega(n^{2})$, $\Omega(m^{\nicefrac{4}{3}})$ & [\cref{thm:p4c4lowerbounds}] \\ \cmidrule(lr){3-4} \cmidrule(lr){5-6} 
		&\multirow{3}{*}{diamond} %
				& $O(m^{\nicefrac{3}{2}})$  	& \cite{EG04}			& $\Theta(m^2)$ & \\
		&		& \multicolumn{2}{c}{?}							& $O(c^2 n^2 + m^{\nicefrac{3}{2}})$ & [\cref{thm:enumh}] \\ 
		&		&						&					& $\Omega(n^{2})$ & \\ \cmidrule(lr){3-4} \cmidrule(lr){5-6} 
		&\multirow{2}{*}{$K_4$} %
				& $O(m^2)$ 				&					& \multirow{2}{*}{$\Theta(m^2)$} & \multirow{2}{*}{[\cref{obs:n^k-subgraphs}]} \\
		&		& $O(cn^3)$ 				& [\cref{thm:findkk}]	&  \\ 
		\bottomrule
	\end{tabular}
\end{table}

\subparagraph*{Further Related Work.}
%
We refer to \Cref{tab:overview} for an overview on prior results on subgraph detection algorithms for three- and four-vertex subgraphs.
As to subgraph enumeration, it is folklore that for a graph $H$ on $k$ vertices, an algorithm that enumerates all induced copies of $H$ takes $\Theta(n^k)$ time (see \Cref{obs:n^k-subgraphs}).
For enumerating triangles, an $O(m^{\nicefrac{3}{2}})$-time algorithm is provided by \citet{IR78}.

As to FPT~in~P, there are a few works on detecting and counting triangles~\cite{KN18,BFNN19,CDP19}.
\citet{KL19} provided parameterized algorithms for several induced subgraph detection problems where the subgraph has four vertices.
Their parameter is the order of the largest clique in the host graph.

Being a relatively new parameter, there is not much work on parameterized algorithms exploiting the $c$-closure~\cite{FRSWW20,KKS20,KKS20b}.
All maximal cliques can be enumerated in $3^{\nicefrac{c}{3}} \cdot n^{O(1)}$ time \cite{FRSWW20}.
For constant $c$, they all showed that there are $O(n^{2 - 2^{1 - c}})$ maximal cliques in $c$-closed graphs, which was previously shown for $c = 2$ \cite{EHSS11}.
Dense subgraphs such as $s$-plexes, $s$-defective cliques, and bicliques can be enumerated in $2^c \cdot n^{O(1)}$ time \cite{KKS20b}.
Moreover, polynomial kernels for several NP-hard graph problems are known \cite{KKS20}.

%

\section{Preliminaries}

For $k \in \mathds{N}$, let $[k]$ denote the set $\{ 1, \dots, k \}$.
Throughout the paper, we use $G$ to denote an undirected graph.
Let $V(G)$ and $E(G)$ be the vertex set and the edge set of $G$, respectively, with~$n = |V(G)|$ and~$m = |E(G)|$.
We will use $\overline{G}$ for the complement of $G$.
For a vertex $v \in V(G)$, let $N(v) = \{ u \mid uv \in E(G) \}$ and $N[v] = N(v) \cup \{ v \}$ denote its open and closed neighborhood, respectively.
The degree of a vertex $v$ is $\deg(v) = |N(v)|$.
A vertex~$v$ is universal if~$\deg(v) = n$.
For a vertex set $S \subseteq V(G)$, the notation $G[S]$ is used for the subgraph induced by $S$.
The path on $k$ vertices is denoted by $P_k$, the complete graph on~$k$ vertices is denoted by $K_k$, and the complete bipartite graph with the parts containing~$k$ and~$\ell$ vertices is denoted by~$K_{k,\ell}$.

\begin{definition}[\cite{FRSWW20}]
	A graph $G$ is \emph{$c$-closed} if $|N(u) \cap N(v)| < c$ for all pairs of nonadjacent vertices $u, v$.
	The \emph{$c$-closure} of $G$ is the smallest integer $c$ such that $G$ is $c$-closed.
\end{definition}

Considering the landscape of graph parameters, the $c$-closure is obviously ``smaller'' than the maximum vertex degree~$\Delta$ of the graph, i.\,e., $c \le \Delta + 1$.
Other (common) parameters smaller than~$\Delta$ are minimum degree, degeneracy, acyclic chromatic number, and $h$-index~\cite{ParamHier19}.
These parameters are unrelated to $c$-closure: they are all~$O(1)$ on a~$K_{2,n-2}$ ($c$-closure is~$n-1$) and large~$\Omega(n)$ on a~$K_n$ ($c$-closure is 1).
These two examples also show that there are graphs with~$c \cdot n \in \Theta(m^2)$ and graphs with~$c \cdot n \in \Theta(\sqrt{m})$.

\section{Three-Vertex Induced Subgraphs}\label{sec:three-vertex}

In this section, we consider the three-vertex induced subgraphs.
We start with the edgeless subgraph.
For constant~$k$, it was shown that finding an independent set~$\co{K_k}$ on~$k$ vertices in~$c$-closed graphs can be found in~$O(m + c^k)$ time~\cite{KKS20}.
Thus, a~$\co{K_3}$ can be detected in~$O(m + c^3)$ time.
Enumerating all~$\co{K_k}$ cannot be done in~$o(n^k)$ time, even on $c$-closed graphs as an edgeless graph on~$n$ vertices is $1$-closed and contains~$\Theta(n^k)$ many~$\co{K_k}$'s.
This settles finding and enumerating~$\co{K_3}$'s.

As a side result, we remark that the $O(m + c^k)$-time algorithm for detecting a~$\co{K_k}$ can be used as subroutine to find stars.
However, note that the subsequent result is not useful for finding stars few leafs as the existing algorithms for finding a~$K_{1,2} = P_3$ and a claw~$K_{1,3}$ on general graphs are faster (see \Cref{tab:overview}).
\begin{theorem}
  \label{thm:findclaw}
  There is an $O(c^k m^{1 - \nicefrac{1}{k}} + m^{2 - \nicefrac{1}{k}})$-time algorithm to find an induced $K_{1, k}$ for constant $k$.
\end{theorem}
\begin{proof}
  Let $D = m^{\nicefrac{1}{k}}$.
  An induced $K_{1, k}$ in which the center is of degree at most $D$ can be found in $O(D^{k - 1} m) = O(m^{2 - \nicefrac{1}{k}})$ time:
  There are $m$ choices for the center and one of its leaves and $O(D^{k - 1})$ choices for the other $k - 1$ leaves.
  Hence, it remains to find a $K_{1, k}$ where the center is a vertex~$v$ of degree at least $D$.
  This can be done by looking for an independent set in~$N(v)$.
  Recall that an independent set of order $k$ can be found in $O(m + c^k)$ time \cite{KKS20}.
  Since there are $O(\nicefrac{m}{D})$ vertices of degree at least $D$, the overall running time is~$O(\nicefrac{m}{D} \cdot (m + c^k)) = O(c^k m^{1 - \nicefrac{1}{k}} + m^{2 - \nicefrac{1}{k}})$.
\end{proof}

With the case of~$\co{K_3}$ being settled, we turn to the remaining three-vertex graphs: $\co{P_3}$, $P_3$, and~$K_3$.
As already mentioned by \citet{WWWY15}, one can find a~$\co{P_3}$ in linear time. 
As we were unable to find the corresponding algorithm in the literature, we provide one for completeness.

\begin{observation}[folklore]
	\label{thm:findcop3}
	There is an~$O(n+m)$-time algorithm to find an induced~$\co{P_3}$.
\end{observation}
\begin{proof}
	Assume that the input graph has at least one edge; otherwise there is no~$\co{P_3}$.
	Further assume that there is no isolated vertex; otherwise the isolated vertex and any edge forms a~$\co{P_3}$.
	Clearly, this assumptions can be checked in linear time.
	
	Partition the vertex into two parts $V_1$ and $V_2$ where~$V_1$ is the set of vertices of degree less than $n/2 - 1$ and $V_2 = V \setminus V_1$. 
	One can check in~$O(m)$ time whether~$V_1$ is an independent set.
	If not, then~$u,v \in V_1$ with~$uv \in E$ is part of a~$\co{P_3}$: Since~$|N[\{ u, v \}] \le n-2$, there is a vertex~$w \in V \setminus N[\{u,v\}]$.
	Thus assume that~$V_1$ is an independent set.
	
	Next consider the case that there is a vertex~$v \in V_1$ with~$N(v) \ne V_2$.
	Thus, there is a vertex~$u \in V_2 \setminus N(v)$.
	Since~$\deg(u) \ge n/2-1$ and~$\deg(v) < n/2 -1$ it follows that there is a vertex~$w \in N(u) \setminus N(v)$.
	Thus~$v,u,w$ forms a~$\co{P_3}$.
	Note that~$v$ and~$u$ can be found (if one exists) in linear time.
	Once~$u$ and~$v$ are fixed, one can find~$w$ with another linear-time scan of the graph.
	
	It remains to consider the case that for each vertex~$v \in V_1$ we have~$N(v) = V_2$.
	Since~$V_1$ is an independent set, it follows that no vertex in~$V_1$ can be part in a~$\co{P_3}$.
	Thus, it suffices to look for a $\co{P_3}$'s in $G[V_2]$.
	Since~$\deg(v) \ge n/2$ for each~$v \in V_2$, it follows that we can compute~$\co{G[V_2]}$ in~$O(n+m)$ time.
	Moreover, we can find a~$P_3$ in~$\co{G[V_2]}$ in~$O(n+m)$ time\footnote{\label{note1}Note that each connected component that is not a clique contains a~$P_3$, which can be found with a simple BFS from a non-universal vertex.}.
\end{proof}

Note that all~$\co{P_3}$'s can be enumerated easily in~$O(nm)$ time: Enumerate all combinations of one edge and one vertex and check whether they induce a~$\co{P_3}$.
The following observation shows that this running time bound is tight, even for $1$-closed graphs.
\begin{observation}[folklore]
	\label{obs:n^k-subgraphs}
	Let~$k$ be a constant. 
	For every subgraph~$H$ on~$k$ vertices there is an $n$-vertex graph~$G$ containing~$\Theta(n^k)$ distinct occurrences of~$H$.
	Moreover, if~$H$ does not contain an induced~$P_3$, then~$G$ is $1$-closed.
\end{observation}

\begin{proof}
	For a fixed~$H = (\{v_1, \dots, v_k\},E)$ on $k$ vertices, let~$G$ be the graph obtained by replacing each vertex with a clique of order $n/k$.
	Formally, let $V(G) = (\{v_i^a \mid i \in [k], a \in [n/k]\}$ and $E(G) = \{ v_i^a v_i^b \mid i \in [k], a, b \in [n/k] \} \cup \{ v_i^a v_j^b \mid v_iv_j \in E, a, b \in [n/k] \}$.
	Clearly, for each combination~$a_1, a_2, \ldots, a_k \in [n/k]$, the graph~$G[\{v_1^{a_1}, v_2^{a_2}, \ldots,v_k^{a_k}\}]$ is isomorphic to~$H$.
	This show the first part.
	
	As to the second part, observe that if~$H$ does not contain an induced~$P_3$, then~$H$ is a cluster graph.
	By construction, it follows that also~$G$ is a cluster graph thus is $1$-closed.
\end{proof}


We continue with~$P_3$'s. 
As a~$P_3$ can be found in linear time\cref{note1}, there is no need to consider $c$-closed graphs for the detection problem.
We thus turn to enumeration. 
First, observe that a start~$K_{1,n}$ contains~$\Theta(n^2)$ many~$P_3$'s and is $2$-closed.
Thus, the following upper bound on the number of~$P_3$'s is tight.

\begin{lemma}
  \label{lem:nump3}
  A $c$-closed graph~$G$ has~$O(cn^2)$ induced~$P_3$'s.
\end{lemma}
\begin{proof}
  Let~$\mathcal{P}_3$ be the set of all~$P_3$'s in~$G$ and let~$\mathcal{P}_3^{uv}$ be the set of all~$P_3$'s with endpoints~$u$ and~$v$.
  By definition, $\mathcal{P}_3 = \bigcup_{uv \notin E(G)} \mathcal{P}_3^{uv}$.
  Since~$G$ is $c$-closed, $|\mathcal{P}_3^{uv}| < c$ for each~$uv \notin E(G)$.
  Thus, we obtain~$|\mathcal{P}_3| = \sum_{uv \notin E(G)} |\mathcal{P}_3^{uv}| < c \binom{n}{2} = O(cn^2)$.
\end{proof}

Ho{\`{a}}ng et al.~\cite{HKSS13} showed that \Cref{algo:p3s} runs in~$O(m + (\# P_3) + (\# K_3))$ time, where~$(\# P_3)$ and~$(\# K_3)$ are the number of~$P_3$'s and~$K_3$'s, respectively:
\begin{algorithm}[t]
  \caption{An algorithm for enumerating~$P_3$'s.}
  \label{algo:p3s}
  \begin{algorithmic}[1]
    \Function{EnumerateP3s}{$G$}
      \ForAll{$uv \in E(G)$} \label{algo:p3:edge}
         \ForAll{$w \in N(u) \cup N(v)$} \label{algo:p3:neighbor}
           \State \algorithmicif\ $w \notin N(u) \cap N(w)$ \algorithmicthen\ \textbf{output} a $P_3 = (u, v, w)$. \label{algo:p3:check} 
         \EndFor
      \EndFor
    \EndFunction
  \end{algorithmic}
\end{algorithm}
The algorithm considers each edge~$uv$ and each vertex~$w$ incident with~$uv$.
Thus, $u$, $v$, and~$w$ form either a~$P_3$ or a triangle.
Since~$(\# K_3) \in O(m^{\nicefrac{3}{2}})$ \cite{IR78}, we obtain the following theorem:

\begin{theorem}
  \label{thm:enump3}
  There is an~$O(cn^2 + m^{\nicefrac{3}{2}})$-time algorithm to enumerate all~$P_3$'s.
\end{theorem}

\citet{FRSWW20} showed that a set of cliques containing all maximal cliques can be enumerated in $O(p_3(n, c) + 3^{\nicefrac{c}{3}} n^2)$ time, where $p_3(n, c)$ is the time complexity to list all induced $P_3$'s.
They noted that $p_3(n, c) = O(c n^{2 + o(1)} + c^{3 - \omega - \alpha} n^\omega + n^\omega \log n)$ due to the result of \citet{GKL09}, where $\omega$ and $\alpha$ are the the matrix multiplication exponent and the dual exponent of matrix multiplication, respectively. 
Using~\Cref{thm:enump3} to bound~$p_3(n, c)$ gives the following:

\begin{corollary}
  There is an $O(m^{\nicefrac{3}{2}} + 3^{\nicefrac{c}{3}} n^2)$-time algorithm to find a set of cliques containing all maximal cliques.
\end{corollary}

We can adapt \Cref{algo:p3s} to an algorithm for finding a triangle:
As mentioned above, we find either a triangle or~$P_3$ in Line~\ref{algo:p3:check} of \Cref{algo:p3s}.
In order to find a triangle in~$O(cn^2)$ time, we just terminate the algorithm as soon as one is detected.

\begin{theorem}
  \label{thm:findk3}
  There is an~$O(cn^2)$-time algorithm to find a triangle.
\end{theorem}

\begin{corollary}
  \label{thm:findkk}
  There is an $O(cn^{k - 1})$-time algorithm to find a clique $K_k$ of size $k$.
\end{corollary}
\begin{proof}
  For each subset $S$ of $k - 3$ vertices, we check whether there is a triangle in $\bigcap_{v \in S} N(v)$ in $O(cn^2)$ time by \Cref{thm:findk3}.
\end{proof}

Next, we develop a more efficient algorithm for finding a triangle in sparse graphs.

\begin{theorem}
  \label{thm:findk3-2}
  There is an~$O(c^{\nicefrac{1}{3}} m^{\nicefrac{4}{3}})$-time algorithm to find a triangle.
\end{theorem}
\begin{proof}
  Let~$D = c^{\nicefrac{1}{3}} m^{\nicefrac{1}{3}}$.
  Let~$V_1$ be the set of vertices with degree at least~$D$ and let~$V_2 = V(G) \setminus V_1$.
  Note that~$|V_1| \in O(\nicefrac{m}{D})$.
  If there is a triangle in~$G[V_1]$, then it can be found in~$O(c (\nicefrac{m}{D})^2) = O(c^{\nicefrac{1}{3}} m^{\nicefrac{4}{3}})$ time by \Cref{thm:findk3}.
  If there is a triangle containing at least one vertex of~$V_2$, then it can be found in~$O(m D) = O(c^{\nicefrac{1}{3}} m^{\nicefrac{4}{3}})$ time.
\end{proof}

As to enumerating triangles, it follows from \Cref{obs:n^k-subgraphs} that the $O(m^{\nicefrac{3}{2}})$-time algorithm of \citet{IR78} cannot be improved even in $1$-closed graphs.

As a side-result, we also show that by enumerating all~$P_3$'s, one can compute the $c$-closure. 

\begin{theorem}
  There is an~$O(cn^2 + m^{\nicefrac{3}{2}})$-time algorithm to compute the $c$-closure. 
\end{theorem}
\begin{proof}
	We first enumerate all~$P_3$'s in~$O(cn^2 + m^{\nicefrac{3}{2}})$ time.
	Let~$\mathcal{P}_3$ be the set of all~$P_3$'s and let~$\mathcal{P}_3^{uv}$ be the set of all~$P_3$'s with endpoints~$u$ and~$v$.
	Once we obtain~$\mathcal{P}_3$ by \Cref{thm:enump3} in~$O(cn^2 + m^{\nicefrac{3}{2}})$ time, one can find all~$\mathcal{P}_3^{uv}$ in~$O(|\mathcal{P}_3|)$ time with a radix sort (recall that~$|\mathcal{P}_3| \in O(cn^2)$ by \Cref{lem:nump3}).
	Then the $c$-closure of~$G$ equals~$\max_{uv \notin E(G)} |\mathcal{P}_3^{uv}| + 1$.
\end{proof}

We remark that deciding whether a graph is 2-closed requires~$O(n m^{\nicefrac{2}{3}})$ time \cite{EHSS11}.

\section{Four-Vertex Induced Subgraphs} \label{sec:four-vertex}
In this section, we consider four-vertex subgraphs.
We turn our attention first to the enumeration aspect and then to the detection part.

\subsection{Enumeration}


\subparagraph*{Algorithms.}
Recall that if our four-vertex subgraph~$H$ does not contain an induced~$P_3$, then \Cref{obs:n^k-subgraphs} excludes algorithms with running time~$O(f(c) n^{4-\varepsilon})$ for any function~$f$ and any~$\varepsilon > 0$.
This applies to five of the eleven subgraphs: co-diamond, co-square, co-claw, and $K_4$.
Interestingly, as we show below, we can have algorithms with running time~$O(cn^3)$ or better for the other six subgraphs, namely, co-paw, $P_4$, claw~$K_{1,3}$, paw, square~$C_4$, and diamond.
This is implied by the next simple but general theorem.
Before stating the theorem, we need some more notation: 
For a graph $H$, let $i_2(H)$ be the minimum size of a vertex set $S$ such that each vertex in $V(H) \setminus S$ has at least two nonadjacent neighbors in~$S$.
For instance, $i_2(C_4) = 2$, $i_2(K_{1, 3}) = 3$, and~$i_2(K_{4}) = i_2(\co{K_{4}}) = 4$.

\begin{theorem}
  \label{thm:enumh}
  Let~$H$ be a graph.
  There is an $O(c^{|V(H)| - i_2(H)} n^{i_2(H)} + cn^2 + m^{\nicefrac{3}{2}})$ time algorithm to enumerate induced copies of $H$.
\end{theorem}
\begin{proof}
  We first compute the set $N(u) \cap N(v)$ for each pair $u, v$ of nonadjacent vertices.
  We can do so by enumerating all $P_3$'s using \Cref{algo:p3s} in $O(c n^2 + m^{\nicefrac{3}{2}})$ time.
  We consider each choice $S$ of $|V(H)| - i_2(G)$ vertices such that each vertex in $V(H) \setminus S$ has at least two nonadjacent vertices.
  For $V(H) \setminus S$, there are $c^{|V(H)| - i_2(H)}$ choices.
\end{proof}

This algorithm can enumerate squares and diamonds in~$O(m^{\nicefrac{3}{2}} + c^2 n^2)$ time and co-paws, $P_4$'s, claws, and paws in~$O(cn^3)$ time.
Note that we can construct $O(1)$-close graphs containing~$\Theta(n^3)$ co-paws, claws, or paws respectively (see discussion in the second part of this subsection).
However, for~$P_4$, square, and diamond we do not have fitting lower bounds (in terms of~$c$ and~$n$).

For $P_4$'s, paws, and squares we found alternative bounds.
As we see in the second part of this subsection the running time of the following algorithm for~$P_4$'s and paws is tight.
\begin{observation}
	\label{obs:enumP4}
	There is an an $O(c n m)$-time algorithm to enumerate all induced~$P_4$'s and all paws.
\end{observation}
\begin{proof}
	By considering all combinations of one edge and one vertex, one fixes three vertices in a~$P_4$ (a paw). 
	For $P_4$'s (paws) assume the non-fixed vertex is one of the two degree-two vertices (the degree-three vertex).
	It follows from the definition of $c$-closure that there are at most~$c-1$ choices for the fourth vertex.
  Note that these choices can be obtained using \Cref{algo:p3s} in $O(cn^2 + m^{\nicefrac{3}{2}})$ time.
	This results in an~$O(cnm)$-time algorithm.
\end{proof}

As we shall in the second part of this section, there are 3-closed graphs with $\Theta(m^{\nicefrac{4}{3}})$ induced copies of $C_4$.
Thus, the running time of the following algorithm could still be improved slightly.
\begin{theorem}
  \label{thm:enumc4}
  There is an an $O(c m^{\nicefrac{3}{2}})$-time algorithm to enumerate all induced squares.
\end{theorem}
\begin{proof}
  Let $D = c^{\nicefrac{1}{2}} m^{\nicefrac{1}{4}}$.
  We call a vertex high-degree if its degree is at least $D$ and low-degree otherwise.
  We consider two cases based on which vertices of the square are high-degree.
  \begin{romanenumerate}
    \item
      Two consecutive vertices are low-degree:
      First, we consider each edge $uv$ where both endpoints $u$ and $v$ are of low-degree.
      Then, we consider each neighbors $u'$ and $v'$ of $u$ and $v$, respectively.
      We list the square $(u, v, v', u')$ if $uv', u'v, u'v' \notin E(G)$.
      This requires $O(D^2 m) = O(c m^{\nicefrac{3}{2}})$ time.
    \item
      Two opposite vertices are high-degree:
      We first enumerate all $P_3$'s where both endpoints are high-degree in $O(c(\nicefrac{m}{D})^2 + m^{\nicefrac{3}{2}})$ time.
      We achieve this by adapting \Cref{algo:p3s}:
      We consider each edge $uv$ where at least one endpoint is high-degree in Line~\ref{algo:p3:edge} instead.
      Without loss of generality, assume that $u$ is high-degree.
      Moreover, we consider each high-degree neighbor $w$ of $v$ in Line~\ref{algo:p3:neighbor} instead.
      Then, this algorithm spends $O(1)$ time for each triangle or $P_3$ whose endpoints are both high-degree.
      Since there are $O(m^{\nicefrac{3}{2}})$ triangles and $O(c (\nicefrac{m}{D})^2)$ $P_3$'s whose endpoints are both high-degree, this adaptation of \Cref{algo:p3s} takes $O(c(\nicefrac{m}{D})^2 + m^{\nicefrac{3}{2}})$ time.
      Thus, we have the set of common neighbors of each pair of nonadjacent high-degree vertices.
      Now we can enumerate all squares where two opposite vertices are high-degree in $O(c^2 (\nicefrac{m}{D})^2) = O(c m^{\nicefrac{3}{2}})$ time.
  \end{romanenumerate}
  Overall, all squares are listed in $O(c m^{\nicefrac{3}{2}})$ time.
\end{proof}

\subparagraph*{(Tight) Lower bounds.}
We now provide (almost) fitting lower bounds. 
Whenever possible, we replace factors of~$n$ by factors of~$m$ (mostly replacing~$n^2$ by~$m$).
This is done via the following simple observation.

\begin{observation}\label{obs:brute-force-matching}
	For a graph $H$ of constant size, there is an $O(n^{|V(H)| - 2\nu(H)} m^{\nu(H)})$-time algorithm to enumerate all induced copies of $H$, where $\nu(H)$ is the maximum matching size of $H$.
\end{observation}
\begin{proof}
  We consider each choice for the set $V'$ of $|V(H)| - 2 \nu(H)$ vertices and the set $E'$ of $\nu(H)$ edges.
  Note that there are $n^{|V(H)| - 2 \nu(H)} m^{\nu(H)}$ such choices.
  Since $H$ of constant size, whether $V' \cup V(E')$ forms an induced $H$ can be checked in constant time.
\end{proof}

\Cref{obs:n^k-subgraphs,obs:brute-force-matching} yield matching running time upper and lower bounds even in $1$-closed graphs for the task of enumerating $\co{K_4}$'s, co-diamonds, co-squares, co-claws, or~$K_4$'s.
The remaining six cases are discussed below (in the order they are listed in \Cref{tab:overview})


Start with a co-paw.
The upper bound~$O(n^2m)$ and~$O(nm^2)$ follow by simple brute force selecting~$i \in [2]$ edges and~$3-i$ vertices (as in \Cref{obs:brute-force-matching}).
As to the lower bound consider the disjoint union of an independent set and a star~$\co{K_{n/2}} + K_{1,n/2-1}$: It is~$2$-closed and has~$m = \Theta(n)$ edges and contains~$\Theta(cn^3) = \Theta(n^3)$ co-paws.


As to~$P_4$'s, observe that again the upper bound~$O(m^2)$ follows from \Cref{obs:brute-force-matching}.
As to the lower bound, consider the graph resulting from making the centers of two~$K_{1,n/2-1}$'s adjacent: it is $2$-closed and contains~$\Theta(n^2) = \Theta(m^2) = \Theta(c n m)$ many~$P_4$'s.
Note that this lower bound fits to the algorithm in \Cref{obs:enumP4} but leaves a gap to the~$O(cn^3)$-time algorithm following from \Cref{thm:enumh}.
Interestingly, we can improve the lower bound as stated in the next theorem, but also the the new lower bound does not match the $O(cn^3)$ upper bound.


\begin{theorem}
  \label{thm:p4c4lowerbounds}
	There is an infinite family of $3$-closed graphs containing $\Theta(n^{\nicefrac{5}{2}})$ $P_4$'s and $\Theta(n^2 + m^{\nicefrac{4}{3}})$ squares.
\end{theorem}
\begin{proof}
	Suppose that $n' = p^2 + p + 1$ for an integer~$p>1$ and consider a projective plane $P$ on $n'$ points and~$n'$ lines.
	It fulfills the following properties:
	\begin{enumerate}
	\item For any pair of points (lines), there is exactly one line incident with both points (points).
	\item Each point (line) is incident with exactly $p + 1$ lines (points).
	\end{enumerate}
	See e.g.\ \citet{AS15} for more on projective planes.

	Now consider the graph $G$ constructed as follows:
	We introduce vertices $u_1, u_2$ for each point $u$ of $P$ and vertices $v_1, v_2$ for each line $v$ of $P$.
	Then, we add an edge $u_1 u_2$ for each point $u$ and $v_1 v_2$ for each line $v$.
	We also add edges $u_1 v_1, u_1 v_2, u_2 v_1, u_2 v_2$ for each pair of point $u$ and line $v$ that are incident in $P$.
	The constructed graph $G$ is 3-closed:
	For two distinct points $u$ and $u'$ (lines $v$ and $v'$), the vertices $u_i$ and $u_j'$ ($v_i$ and $v_j'$) for $i, j \in [2]$ have exactly two common neighbors by the first property of projective planes.
	For a non-incident pair of a point $u$ and a line $v$, the vertices $u_i$ and $v_j$ for $i, j \in [2]$ have no common neighbor.
	Moreover, $G$ has $n = 4n'$ vertices and $m = 4pn' + 6n' \in \Theta(n^{\nicefrac{3}{2}})$ edges by the second property of projective planes.

  Finally, we count the number of $P_4$'s and squares.
  We begin with $P_4$'s.
  Let $u$ and $u'$ be distinct points.
  By the properties of projective planes, there is exactly one line $v$ on which both $u$ and $u'$ lie and there are exactly $p$ lines that are incident with $u$ and not with~$u'$.
  Let $v'$ be one of these $p$ lines.
  Observe that $(u_1, v_1, u_1', v_1')$ is a $P_4$ in $G$.
  Hence, $G$ has $\Theta(pn^2) = \Theta(n^{\nicefrac{5}{2}}) = \Theta(m^{\nicefrac{5}{3}})$ $P_4$'s.
  Next, we consider squares.
	For each pair of distinct points $u$ and $u'$, there exists an induced $C_4$ on $(u_1, v_1, u_1', v_2)$, where $v$ is the line incident with both $u$ and $u'$.
	Thus, there are $\Theta(n^2) = \Theta(m^{\nicefrac{4}{3}})$ squares in $G$.
\end{proof}

We remark that a construction similar to the above one was used to show a lower bound on the number of maximal cliques in 2-closed graphs by \citet{EHSS11}.

Continuing with claws, observe that the upper bound~$O(n^2m)$ follows from \Cref{obs:brute-force-matching}.
As to the lower bound, consider a star~$K_{1,n-1}$: it is $2$-closed and contains~$\Theta(c n^3) = \Theta(n^2m)$ claws.


As to paws, observe that again~$O(m^2)$ time follows from \Cref{obs:brute-force-matching}.
For the lower bound, consider a clique~$K_{n/2}$ where at one vertex of the clique there are $n/2$ degree-one vertices attached.
This results in a $2$-closed graph with~$\Theta(c n^3) = \Theta(n^3)$ paws.

Next, consider squares.
Again, $O(m^2)$ time follows from \Cref{obs:brute-force-matching}.
A lower bound is provided in \Cref{thm:p4c4lowerbounds}.
Note that it does not match the upper bound~$O(m^{\nicefrac{3}{2}} + c^2 n^2)$ that follows from \Cref{thm:enumh}.

Finally, consider diamonds. 
Observe that again $O(m^2)$ time follows from \Cref{obs:brute-force-matching}.
As for lower bounds, consider a graph obtained by making the two high-degree vertices in a~$K_{2,n-2}$ adjacent.
This graph is $3$-closed and has~$\Theta(c^2n^2) = \Theta(n^2) = \Theta(m^2)$ diamonds: combining the two high-degree vertices with any two independent-set vertices form a diamond.
Note that the algorithm following from \Cref{thm:enumh} has an additional~$O(m^{\nicefrac{3}{2}})$ term in its running time which means it not tight.


\subsection{Detection}

In this section, we provide efficient algorithms for five out of the eleven induced subgraph detection problems on $c$-closed graphs, namely for the subgraphs co-diamond, co-paw, co-square, paw, and square.
Note that a~$P_4$ can be found in linear time~\cite{CPS85}, thus there is no room for improvement.
For~$\co{K_4}$ a faster algorithm on $c$-closed graphs is known~\cite{KKS20} (see also first paragraph of \Cref{sec:three-vertex}).
Hence, for four subgraphs, namely claw, co-claw, diamond, and~$K_4$, the question for fast algorithms on $c$-closed graphs remain open.

Out of the five positive results, detecting a co-diamond and a co-square require new algorithms. 
For the remaining three subgraph, the results either directly from \Cref{thm:enumc4} (for square) or from known characterizations via induced three-vertex subgraphs and results from \Cref{sec:three-vertex} (for co-paw and paw).
We start by briefly discussing the latter (co-paw and paw).
Afterwards, we show the algorithms for detecting a co-diamond and a co-square.
Finally, we provide a algorithm for detecting a diamond in a gem-free $c$-closed graph.
Moreover, we highlight the issue that needs to be resolved in order to remove the gem-free assumption. 

\subparagraph*{Co-paw and paw.}
We use the characterization of \citet{Ola88}: A graph is a paw-free if and only if it is triangle-free or $\overline{P_3}$-free.
Thus, we immediately obtain the following corollary from \Cref{thm:findcop3} and \Cref{thm:findk3,thm:findk3-2}. \todo{@TK: Can we also to find the co-pow / paw?}

\begin{corollary}
  \label{cor:findpaw}
  There is an $O(\min \{ cn^2, c^{\nicefrac{1}{3}} m^{\nicefrac{4}{3}} \})$-time algorithm to detect an induced paw.
\end{corollary}

\begin{corollary}
  \label{cor:findcopaw}
  There is an $O(m + c^3)$-time algorithm to detect an induced co-paw.
\end{corollary}
\begin{proof}
	Due to the result of \citet{Ola88}, it suffices to find a~$P_3$ and~$\co{K_3}$ in the input graph~$G$:
	\begin{align*}
		G \text{ is co-paw-free} \iff \co{G} \text{ is paw-free} & \overset{\text{\cite{Ola88}}}{\iff} \co{G} \text{ is } K_3\text{-free or } \co{P_3}\text{-free} \\ & \iff G \text{ is } \co{K_3}\text{-free or } P_3\text{-free.}
	\end{align*}
	An induced~$P_3$ can be found in $O(n+m)$ time and an independent set of order three can be found in~$O(m + c^3)$ time~\cite{KKS20}.
\end{proof}



\subparagraph*{Co-diamond.} 
We next present our algorithm detecting co-diamonds, which is based on the following structural statements.

\begin{lemma}
	\label{lemma:codiamond}
	If there is a maximal clique~$C$ of order at least~$2c$ in~$G$, then either~$V(G) \setminus C$ is a clique or~$G$ contains a co-diamond.
\end{lemma}
\begin{proof}
	If~$V(G) \setminus C$ is a clique, then clearly the graph~$G$ is co-diamond-free.
	It remains to show that if~$V(G) \setminus C$ is not a clique, then~$G$ contains a diamond.
	To this end, let~$uv \notin E(G)$ for~$u, v \in V(G) \setminus C$.
	Since~$C$ is maximal, there exist vertices~$u', v' \in C$ such that~$uu', vv' \notin E(G)$.
	By the $c$-closure, we have that~$|N(u) \cap C| < c$ and that~$|N(v) \cap C| < c$.
	Therefore, $|C \setminus N[\{ u, v \}]| \ge 2$.
	For~$w, w' \in C \setminus N[\{ u, v \}]$, the four vertices~$(u, v, w, w')$ forms an induced co-diamond.
\end{proof}

In our algorithm, we will use the following statement, which is a small reformulation of \Cref{lemma:codiamond}.

\begin{lemma}
	\label{lem:codiamond2}
	Let~$G$ be a graph that cannot be partitioned into two cliques.
	If there is a clique~$C$ of order at least~$2c$ in~$G$, then $G$ contains a co-diamond.\todo{@TK: Can one find the co-diamond in linear time?!}
\end{lemma}
\begin{proof}
	If~$C$ is maximal, then the statement directly follows from \Cref{lemma:codiamond}.
	Otherwise, let~$C'$ be a maximal clique containing~$C$.
	Clearly, $C'$ is of order at least~$2c$.
	Moreover, since~$V(G)$ cannot be partitioned into two cliques, it follows that~$V(G) \setminus C'$ is not a clique.
	Thus, the statement again follows from \Cref{lemma:codiamond}.
\end{proof}

\begin{theorem}
	\label{thm:codiamond}
	There is an~$O(m + c^2 n)$-time algorithm to detect an induced co-diamond.
\end{theorem}
\begin{proof}
	If~$n \le 6c$, then we can determine whether the input graph~$G$ has an induced co-diamond in~$O(c^3)$ time, using the~$O(m^{\nicefrac{3}{2}})$-time algorithm of~\citet{EG04}.
	So assume that~$n \ge 6c$.

	Then, we determine whether the vertex set~$V(G)$ can be partitioned into two cliques~$C_1$ and~$C_2$.
	If~$m < \binom{n/2}{2}$, then this is impossible (at least one clique needs to be of order~$n/2$).
	Thus, assume~$m \ge \binom{n/2}{2} \in \Theta(n^2)$.
	Hence, in~$O(m)$ time we can simply check whether the complement~$\co{G}$ of~$G$ is bipartite.
	Suppose that there are two cliques~$C_1$ and~$C_2$ such that~$C_1 \cup C_2 = V(G)$.
	Then, we can conclude that~$G$ has no induced co-diamond.
	Thus, we assume in the following that~$V(G)$ cannot be partitioned into two cliques (note that this allows us to invoke \Cref{lem:codiamond2}).

	We claim that if there is an edge~$uw$ such that~$\deg(u) \le 2c$ and~$\deg(w) \le 2c$, then~$G$ has an induced co-diamond.
	Since~$n \ge 6c$, we have~$|V(G) \setminus N[\{ u, w \}]| \ge n - 4c \ge 2c$.
	If~$V(G) \setminus N[\{ u, w \}]$ is a clique (which can be checked in~$O(m)$ time), then by \Cref{lem:codiamond2} there is a co-diamond in~$G$.
	Hence, assume there exist nonadjacent vertices~$v, v' \in V(G) \setminus N[\{ u, w \}]$.
	However, then~$(u, w, v, v')$ forms an induced co-diamond.
	Note that such an induced co-diamond can be found in~$O(m)$ time.

	Next, consider the case that there is a vertex~$v$ such that~$2c < \deg(v) \le n - 2c$.
	We claim that~$G$ has an induced co-diamond in this case.
	Note that~$|V(G) \setminus N[v]| \ge 2c$.
	Hence, if~$V(G) \setminus N[v]$ is a clique, then by \Cref{lem:codiamond2} there is a co-diamond in~$G$.
	Otherwise, there exist nonadjacent vertices~$u, w \in V(G) \setminus N[v]$.
	Moreover, there exists a vertex~$v' \in N(v)$ that is adjacent to neither~$u$ nor~$w$:
	The $c$-closure of~$G$ yields that~$|N(v) \setminus N(\{ u, w \})| \ge |N(v)| - |N(v) \cap N(u)| - |N(v) \cap N(w)| > 2c - 2(c - 1) > 0$.
	Thus, we find an induced co-diamond~$(u, w, v, v')$.
	Hence, we assume in the following that each vertex has a degree of at most~$2c$ or at least~$n-2c+1$.

	It remains to consider the case that each edge contains a vertex of degree at least~$n-2c+1$.
	We iterate over all these high-degree vertex; let~$v$ be such a vertex of degree at least~$n - 2c + 1$.
	To find an co-diamond where~$v$ is one of its degree-one vertices, we simply check whether there are nonadjacent pair~$u, w$ of vertices in~$V(G) \setminus N[v]$.
	Since~$|V(G) \setminus N[v]| < 2c$, we can find~$u$ and~$w$ (if existing) in~$O(c^2)$ time.
	If there is no such pair, then we can conclude that~$G$ has no induced co-diamond containing~$v$.
	Otherwise, there is an induced co-diamond~$(u, w, v, v')$ for~$v' \in N(v) \setminus N[\{ u, w \}]$.
	Note that~$N(v) \setminus N[\{ u, v \}] \ne \emptyset$ by the same argument above.
	Since we spend~$O(c^2)$ time for each vertex of degree at least~$n - 2c$, this step requires~$O(c^2 n)$ time.
\end{proof}

\subparagraph*{Co-square.}
We now consider co-squares. 
The next lemma plays an important role in our co-square detection algorithms.

\begin{lemma}
  \label{lemma:cosquare}
  Suppose that there are vertices $u, v \in V(G)$ such that $\deg(u) \ge c$, $\deg(v) \ge 2c - 1$, and $uv \notin E(G)$.
  Then, $G$ contains an induced co-square.
\end{lemma}
\begin{proof}
  Since $G$ is $c$-closed, we have $|N(u) \cap N(v)| < c$.
  It follows that $|N(u) \setminus N(v)| = |N(u)| - |N(u) \cap N(v)| > 0$.
  Let $w$ be an arbitrary vertex in $N(u) \setminus N(v)$.
  Since $vw \notin E(G)$, we have $|N(v) \cap N(w)| < c$ by the $c$-closure of $G$.
  Thus, $|N(v) \setminus N(\{ u, w \})| \ge |N(v)| - |N(u) \cap N(v)| - |N(v) \cap N(w)| > 0$.
  For an arbitrary vertex $v' \in N(v) \setminus N(\{ u, v \})$, the vertices $(u, w, v, v')$ form an induced co-square.
\end{proof}

We say that a connected component is \emph{trivial} if it consists of one vertex.

\begin{theorem}
  \label{thm:cosquare}
  There is an~$O(m + c^{\nicefrac{44}{7}})$-time randomized algorithm to detect an induced co-square.
\end{theorem}
\begin{proof}
  Let~$C$ be the set of vertices of degree at least~$2c - 1$.
  If~$C$ is not a clique, then~$G$ contains an induced co-square by \Cref{lemma:cosquare}.
  So assume that~$C$ is a clique.
  Let~$S_1, \dots, S_\ell$ be the connected components of~$G - C$.
  If all components are trivial, then there is no induced co-square.
  Moreover, if there are more than one non-trivial connected component, we find an induced co-square.
  Thus, we assume that there is exactly one connected component~$S$ with at least one edge.

  If there is no co-square in~$S$, then the diameter of~$G[S]$ is at most three.
  Since every vertex in~$S$ has degree at most~$2c - 2$, we can assume that~$|S| \in O(c^3)$.
  Furthermore, if~$|C| \ge 2c$, then there is an induced co-square.
  For each vertex~$v \in S$, there exists a vertex in~$C$ that is not adjacent to~$v$, because~$|C| \ge 2c$ and~$v$ has at most~$2c - 2$ neighbors.
  Thus, the~$c$-closure of~$G$ yields that each vertex~$v \in S$ has at most~$c - 1$ neighbors in~$C$.
  Consequently, for an edge~$vv' \in E(G[S])$, there are two vertices~$u, w \in C$ that are not adjacent to~$v$ or~$v'$, which form a co-square along with~$v$ and~$v'$.
  Hence, we can also assume that~$|C| \le 2c$.

  Let~$I = V(G) \setminus (C \cup S)$ be the set of isolated vertices in~$G - C$.
  Now we describe an~$O(m + c^5)$-time algorithm to find a co-square containing a vertex of~$I$.
  For each vertex~$v \in C$, we check whether~$v$ has a neighbor in~$I$.
  This can be done in~$O(m)$ time.
  Then, for each vertex~$v \in C$ with at least one neighbor in~$I$, we check whether there is an edge~$uw \in E(G[S])$ such that~$u, w \notin N(v)$.
  If there is such an edge~$uw$, then~$(v, v', u, w)$ is a co-square, where~$v' \in I$ is a vertex adjacent to~$v$.
  Otherwise, we can conclude that there is no co-square containing a vertex of~$I$.
  Note that this procedure takes~$O(c^5)$ time, because~$|C| \le 2c$ and~$|E(G[S])| \in O(c^4)$.

  Finally, it remains to find a co-square in~$G[C \cup S]$.
  Using the~$O(m^{\nicefrac{11}{7}})$-time algorithm of \citet{WWWY15}, this can be done in~$O(c^{\nicefrac{44}{7}})$.
\end{proof}



\subparagraph*{Diamond.}
The following characterization for diamond-free graphs is used in algorithms that find a diamond and run in time~$O(m^{\nicefrac{3}{2}})$~\cite{EG04} and~$O(\Delta m)$~\cite{KKM00} respectively.
\begin{lemma}[{\cite[Lemma 3]{KKM00}}]
  \label{lemma:diamond}
  A graph $G$ is diamond-free if and only if $G[N(v)]$ is $P_3$-free for each vertex $v \in V(G)$.
\end{lemma}

Using this characterization, we show that if the input graph does not contain any induced gem (five-vertex graph formed by $P_4$ and an additional universal vertex), then an induced diamond can be detected in $O(cn^2)$ time.

\begin{theorem}
  There is an $O(cn^2)$-time algorithm to detect an induced diamond in gem-free graphs.
\end{theorem}
\begin{proof}
  We show that for each vertex $v \in V(G)$, one can verify in $O(cn)$ time whether $G[N(v)]$ is $P_3$-free.
  By \Cref{lemma:diamond}, this yields an $O(cn^2)$-time algorithm for finding a diamond.
  First, we use an $O(cn)$-time algorithm (\Cref{algo:diamond}) that determines that either
  \begin{itemize}
    \item
      $G[N(v)]$ is not $P_3$-free, or
    \item
      there is an inclusion-maximal independent set $I \subseteq N(v)$ in $G[N(v)]$ such that $N(u) \cap N(w) \cap N(v) = \emptyset$ for all $u, w \in I$.
  \end{itemize}
  \begin{algorithm}[t]
    \caption{An algorithm for finding a $P_3$ or an independent set in the neighborhood of $v$.}
    \label{algo:diamond}
    \begin{algorithmic}[1]
      \Function{FindIS}{$G, v$}
        \State $I \leftarrow \emptyset$, $J \leftarrow N(v)$.
        \While{$J \ne \emptyset$} \Comment{Each vertex in $N(v) \setminus J$ has at least one neighbor in $I$.} \label{line:diamond:loop}
          \State Let $u$ be an arbitrary vertex in $J$.
          \State $I \leftarrow I \cup \{ u \}$.
          \ForAll{$w \in N(u)$}
            \State \algorithmicif\ $w \notin N(v)$ \algorithmicthen\ \Continue.
            \State \algorithmicif\ $w \notin J$ \algorithmicthen\ there is a $P_3$ in $G[N(v)]$; \Return \label{algo:diamond:p3}
          \EndFor
          \State $J \leftarrow J \setminus N[u]$
        \EndWhile
        \State \Return $I$.
      \EndFunction
    \end{algorithmic}
  \end{algorithm}
  Basically, \Cref{algo:diamond} keeps adding some vertex $u \in J$ to an independent set $I$, until there is no vertex left in $J$. 
  In doing so, \Cref{algo:diamond} removes neighbors of $u$ from $J$.
  Suppose that $u$ has a neighbor $w$ not in $J$ (Line~\ref{algo:diamond:p3}).
  Then, there is a $P_3$, because $w$ has a neighbor in $I \setminus \{ u \}$.
  We emphasize that \Cref{algo:diamond} only requires the adjacency list of $G$;
  we avoid constructing the adjacency list of $G[N(v)]$, since it could take $\Omega(m)$ time.

  We show that the time complexity of \Cref{algo:diamond} is $O(cn)$.
  In particular, consider the case \Cref{algo:diamond} finds an independent set $I$ (the proof is analogous for the other case \Cref{algo:diamond} finds a $P_3$).
  For each vertices $u \in I$ and $w \in N(u)$, \Cref{algo:diamond} spends $O(1)$ time.
  It clearly holds that $|\bigcup_{u \in I} N(u)| = |\bigcup_{u \in I} N(u) \cap N(v)| + |\bigcup_{u \in I} N(u) \setminus N(v)|$.
  Since $N(u) \cap N(u') \cap N(v) = \emptyset$ for all $u, u' \in I$ by the choice of $I$, we have $|\bigcup_{u \in I} N(u) \cap N(v)| \le \deg(v) \in O(n)$.
  Moreover, we have $|\bigcup_{u \in I} N(u) \setminus N(v)| < cn$:
  To see why, note that there is a $P_3$ on $(v, u, w)$ for each vertex $w \in N(u) \setminus N(v)$.
  For each choice of $w \notin N(v)$, the vertices $v$ and $w$ have at most $c - 1$ neighbors in common, because $G$ is $c$-closed.
  Thus, $|\bigcup_{u \in I} N(u)| \in O(cn)$ and, therefore, \Cref{algo:diamond} executes in $O(cn)$ time.

  If \Cref{algo:diamond} finds a $P_3$ in the neighborhood of $v$, then we immediately find a diamond in $G$.
  Thus, assume that \Cref{algo:diamond} returns an independent set $I$.
  Let $S_u = N(u) \cap N(v)$ for each $u \in I$.
  Note that $S_u$ is pairwise disjoint for all $N_u$ by the construction of $I$.
  Since $G$ is gem-free, there is no edge $ww' \in E(G)$ such that $w \in S_u$ and $w' \in S_{u'}$ for $u \ne u' \in I$.
  Thus, in order to decide whether there is a~$P_3$ in~$N(v)$ it remains to decide whether $S_u$ is a clique or not for all $u \in I$.
  Let $I_1 = \{ u \in I \mid |S_u| \le c \}$ and let $I_2 = I \setminus I_1$.
  For each $u \in I_1$, we verify whether all the vertices in $S_u$ are pairwise adjacent.
  This takes $\sum_{u \in I_1} |S_u|^2 \le c \sum_{u \in I_1} |S_u| \le c \cdot \deg(v) \in O(cn)$ time.
  For each $u \in I_2$, we do as follows:
  Let $T \subseteq S_v$ be an arbitrary subset of exactly $c$ vertices.
  We verify that the vertices in $T$ are pairwise adjacent in $O(c^2)$ time.
  Then, we verify whether $w$ and $w'$ are adjacent for each vertices $w \in T$ and $w' \in T \setminus S_v$.
  If $ww' \notin E(G)$ for some $w \in T$ and $w' \in T \setminus S_v$, then $(w, u, w')$ forms a $P_3$.
  Otherwise, $S_v$ is a clique:
  By the $c$-closure of $G$, any pair of vertices with at least $c$ common neighbors are necessarily adjacent.
  It is easy to see that this procedure takes $O(c^2 + c \cdot |S_u|)$ time for each $u \in I_2$.
  Since $\deg(v) \ge c \cdot |I_2|$, its overall running time is $O(\sum_{u \in I_2} c^2 + \sum_{u \in I_2} c \cdot |S_u|) = O(\nicefrac{\deg(v)}{c} \cdot c^2 + c \cdot \deg(v)) = O(cn)$.
\end{proof}

Unfortunately, we were unable to get rid of the condition of the input graph being gem-free in the above theorem.
One obstacle for this the following: We do not know whether one can compute the common neighborhood for each pair of non-adjacent vertices in~$O(cn^2)$ time, despite the output being of size~$O(cn^2)$.
If this could be done, then there would be easy~$O(cn^2)$ and~$O(c^2n^2)$-time brute-force algorithms for enumerating all~$P_3$'s and diamonds respectively.

\section{Conclusion}

We provided a first systematic study of detecting and enumerating small subgraphs in a given $c$-closed host graph.
While we provide several upper and lower bounds, there remain a couple of open questions (see question marks in \Cref{tab:overview}).
Probably the most important one is whether the common neighborhood for each pair of non-adjacent vertices can be computed in~$O(cn^2)$ time.
A positive answer would immediately improve several of our results.
In particular, it would provide tight algorithms for enumerating~$P_3$'s, squares, and diamonds.
Moreover, it would give first subcubic-time algorithms for detecting diamonds in $c$-closed graph.
Investigating parameterized algorithms for the problems studied in this work with respect to the weak $c$-closure \cite{FRSWW20} (a related but smaller parameter than $c$-closure) is a further task for future research.
Finally, besides detecting and enumerating subgraphs, the task of counting subgraphs is also relevant and not studied so far for $c$-closed host graphs.

\bibliographystyle{abbrvnat}
\bibliography{ref}

\end{document}